\documentclass[reqno]{amsart}
\usepackage{tikz,amssymb,verbatim}

\newcommand{\Reals}{\mathbb{R}}
\newcommand{\C}{\mathbb{C}}
\newcommand{\Spec}{\operatorname{Spec}}

\newcommand{\abs}[1]{\left\vert #1 \right\vert}
\newcommand{\conj}[1]{\overline{#1}}
\newcommand{\norm}[1]{\left\Vert #1 \right\Vert}

\newcommand{\Res}{\mathop{\mathrm{Res}}}
\renewcommand{\Im}{\operatorname{Im}}
\renewcommand{\Re}{\operatorname{Re}}

\newtheorem{theorem}{Theorem}[section]
\newtheorem{condition}[theorem]{Hypothesis}
\newtheorem{corollary}[theorem]{Corollary}
\newtheorem{definition}[theorem]{Definition}
\newtheorem{proposition}[theorem]{Proposition}
\newtheorem{remark}[theorem]{Remark}
\newtheorem{lemma}[theorem]{Lemma}

\numberwithin{equation}{section}

\begin{document}
\title[Inverse Scattering Procedure]{On the determinant formula in the inverse scattering procedure with a partially known steplike potential}
\author{Odile Bastille, Alexei Rybkin}
\address{University of Alaska Fairbanks}
\date{July, 2011}
\address{Department of Mathematics and Statistics \\
University of Alaska Fairbanks\\
PO Box 756660\\
Fairbanks, AK 99775}
\email{orbastille@alaska.edu \\arybkin@alaska.edu}
\thanks{Based on research supported in part by the NSF under grant DMS
1009673.}

\begin{abstract} We are concerned with the inverse scattering problem for the full line Schr\"odinger operator $-\partial_x^2+q(x)$ with a steplike potential $q$ a priori known on $\Reals_+=(0,\infty)$. Assuming $q|_{\Reals_+}$ is known and short range, we show that the unknown part $q|_{\Reals_-}$ of $q$ can be recovered by
\begin{equation*}
q|_{\Reals_-}(x)=-2\partial_x^2\log\det\left(1+(1+\mathbb{M}_x^+)^{-1}\mathbb{G}_x\right), 
\end{equation*} 
where $\mathbb{M}_x^+$ is the classical Marchenko operator associated to $q|_{\Reals_+}$ and $\mathbb{G}_x$ is a trace class integral Hankel operator. The kernel of $\mathbb{G}_x$ is explicitly constructed in term of the difference of two suitably defined reflection coefficients. Since $q|_{\Reals_-}$ is not assumed  to have any pattern of behavior at $-\infty$, defining and analyzing scattering quantities becomes a serious issue. Our analysis is based upon some subtle properties of the Titchmarsh-Weyl $m$-function associated with $\Reals_-$.
\end{abstract}

\maketitle
\tableofcontents

\section{Introduction}

As the title suggests, we are concerned with the recovery of an unknown potential of the full line Schr\"{o}dinger operator $-\partial_x^2+q(x)$ from a certain set of scattering data. In its classical formulation, the scattering data consist of the reflection coefficient $R$, bound states $\{-\varkappa_n^2\}_{n=1}^N$ and their norming constants $\{c_n\}_{n=1}^N$.

It is well-known (see e.g. \cite{Deift79}) that the inverse problem
\begin{equation}\label{eq1.1}
S:= \left\{ R(k),-\varkappa_n^2,c_n\right\}_{k\in\Reals,1 \le n\le N} \quad \longrightarrow \quad q(x)
\end{equation}
for all $x$ is uniquely solvable through the famous Faddeev-Marchenko inverse scattering procedure. In fact, there is an explicit formula, referred to sometimes as Bargmann or Dyson, (see e.g. \cite{Fad})
\begin{equation}\label{eq1.2}
q(x)=-2\partial_x^2\log\det (1+\mathbb{M}_x),
\end{equation}
where $\mathbb{M}_x$ is the so-called Marchenko operator, a Hankel integral operator whose kernel is constructed in terms of $S$.

In practice, however, norming constants $\{c_n\}$ are not available. Moreover the authors are unaware of their physical meaning either. Notice that the inverse problem
\begin{equation}\label{eq2.1}
\left\{ R(k)\right\}_{k\in\Reals} \quad \longrightarrow\quad q(x)
\end{equation}
is solved uniquely only if the operator $-\partial_x^2+q(x)$ has no (negative) bound states. The classical example of the so-called one-soliton potential
\begin{equation*}
q(x)=-2\varkappa^2 \operatorname{sech}^2 \left( \varkappa x+ \log \frac{\sqrt{2\varkappa}}{c}\right) 
\end{equation*}
produces the scattering data
\begin{equation*}
\left\{ 0,-\varkappa^2,c\right\}
\end{equation*}
with $R(k)=0$ for all $k$, suggesting that the inverse problem \eqref{eq2.1} is ill-posed in general.

The inverse problem \eqref{eq1.1} was originally posed and solved for $q$'s decaying sufficiently fast at infinity. The complete treatment of this problem is done in \cite{Deift79} under the Faddeev condition\footnote{Certain aspects of the theory developed in \cite{Deift79} actually require the stronger condition
$$\int_\Reals \left(1+\abs{x}^2\right)\abs{q(x)}dx<\infty .$$} 
\begin{equation}\label{eq2.2}
\int_\Reals \left(1+\abs{x}\right)\abs{q(x)}dx<\infty.
\end{equation}

While the Faddeev condition \eqref{eq2.2} is typically satisfied in nuclear physics, many interesting inverse problems (e.g. in geophysics) deal with potentials which do not decay at one of the infinities but still decay at the other infinity. Such potentials are commonly called steplike. A suitable analog of inverse problem \eqref{eq1.1} is well-posed as well (see e.g. \cite{AK01} and the literature cited therein). New circumstances arise, however, due to a richer negative spectrum of $-\partial_x^2+q(x)$. But the classical Marchenko theory can be successfully adjusted to such setting too.

In \cite{BusFo} the inverse problem was solved for $q(x)\to C$, a nonzero constant, as $x\to-\infty$ and $q(x)\to0$ as $x\to\infty$ sufficiently fast (with some gaps fixed in \cite{CoKap}). The case of $q$'s periodic on the, say, left hand side and decaying on the other have been treated first by Hruslov \cite{Hruslov76} (see also \cite{AK01}).

The more general case of $q$'s with no certain pattern of behavior at $-\infty$ and identically zero on the right hand side has been recently treated by one of the authors in \cite{RybIP09}. To the best of our knowledge, in the context of steplike potentials, the determinant formula \eqref{eq1.2} is not available in the literature.

The situation with steplike potentials is similar to the case \eqref{eq2.2} in that the reflection coefficient alone does not determine the potential uniquely. It is natural to ask what can make up for the unavailable data in \eqref{eq1.1} related to the negative spectrum? This problem has generated considerable interest (see e.g. \cite{Akt96,AKM93,BSL95,GW95,RS94}). In the context of Faddeev potentials, the inverse problem
\begin{equation}\label{eq4.1}
\left\{ R(k),q(x)\right\}_{k\in\Reals, x\ge0} \quad \longrightarrow\quad q(x) \;,\; x<0
\end{equation}
is shown to be well-posed. The problem \eqref{eq4.1} is referred to as the inverse scattering problem with partial information on the potential. In \cite{GS97}, the inverse problem \eqref{eq4.1} is solved for essentially arbitrary potentials as long as the reflection coefficient can be suitably defined. The actual procedure of solving \eqref{eq4.1} in \cite{GS97} is not scattering but rather spectral through solving the Gelfand-Levitan integral equation.

In \cite{RybIP09}, one of the authors found a way to adapt the classical Marchenko inverse scattering algorithm to solve \eqref{eq4.1}. No analog of the determinant formula \eqref{eq1.2} appears to be found for inverse problems with partial information on the potential \eqref{eq1.1} in the context of steplike potentials. The present paper intends to deal with this issue. More specifically, we show that if $q$ is locally square integrable,
\begin{equation}\label{eq5.1}
\sup_{x\le0}\int_{x-1}^x \abs{q}^2 <\infty,
\end{equation}
and the known part $q_+=q|_{(0,\infty)}$ is subject to \eqref{eq2.2} then
\begin{equation}\label{eq5.2}
q(x)=-2\partial_x^2\log\det\left( 1+(1+\mathbb{M}_x^+)^{-1}\mathbb{G}_x\right) \quad,\quad x<0.
\end{equation} 
Here $\mathbb{M}_x^+$ is the Marchenko operator constructed in terms of the scattering data for $q_+$ and $\mathbb{G}_x$ is the integral Hankel operator with the kernel expressed in terms of the difference $R-R_+$ of the (right) reflection coefficients $R$ for the whole potential $q$ and $R_+$ corresponding to $q_+$.

We emphasize that the main issue here is the existence of the determinant in \eqref{eq5.2} in the classical Fredholm sense. We prove that the operator $\left(1+\mathbb{M}_x^+\right)^{-1}\mathbb{G}_x$ in \eqref{eq5.2} is trace class for every $x$. Our arguments are based upon the simple but important observation that the difference $R-R_+$ is an analytic function (even though $R$ and $R_+$ are not) and certain limiting procedures (which, as frequently happens in similar situations, are quite involved).

From the geophysical point of view our situation is related to reflection seismology where one is concerned with recovering certain properties $q$ of deeper layers of the Earth given already known properties $q_+$ of shallower layers and measured reflections $R$. Of course if the media do not tend to be homogeneous as the depth increases, a larger range of frequencies is required to investigate deeper layers.

Another example of the inverse problem with partially known steplike potential is related to neutron reflectometry (see e.g. \cite{AS00}) where properties of an unknown material are studied by scattering neutrons and measuring their reflection. The role of $q_+$ is played by a layer of known material, called a coating.

The paper is organized as follows. Section 2 lists our notation. Section 3 introduces the scattering quantities in our setting. Section 4 details relevant properties of the Titchmarsh-Weyl $m$-function which are then applied in our context in Section 5. Section 6 and 7 respectively deal with trace class and Marchenko operators and Section 8 contains the main result.

\section{Notation}

We adhere to standard terminology from analysis, namely $\mathbb{R}%
_{\pm }:=[0,\pm \infty )$, 
$\mathbb{C}$ is the complex plane, 
\begin{equation*}
\mathbb{C}^{+}=\left\{ z\in \mathbb{C}\;:\; \Im z>0\right\} \quad,\quad i\Reals_+=\left\{z \in \C\;:\;z=iy\;,\;y\in\Reals_+\right\},
\end{equation*}
in the upper half plane,
\begin{equation*}
\Reals+ih=\{z\in\C\; :\; z=x+ih\; ,\; x\in\Reals\} 
\end{equation*}
is the real line shifted $h$ units up.
$\norm{\cdot}_{X}$ stands for a norm in a Banach
(Hilbert) space $X$. We use ($S \subseteq \Reals$ and $S$ will typically be $\mathbb{R}$ or $\mathbb{R}_{\pm }$):
\begin{itemize} 
\item the usual Lebesgue spaces 
\begin{align*}
L^{p}\left( S\right) & :=\left\{ f:\norm{f}_{L^{p}\left(
S\right) } := \left( \int_{S}\abs{f\left( x\right)}
^{p}dx\right) ^{1/p}<\infty \right\} \quad,\quad 1\leq
p<\infty \\
L^{\infty }\left( S\right) & :=\left\{ f:\norm{f}_{L^{\infty
}\left( S\right) }:= \operatorname{ess sup}_{x\in S}\abs{f\left(
x\right)} <\infty \right\} , \\
L_{loc}^{p}\left( S\right) & :=\left\{ \cap L^{p}\left( \Delta \right) :\Delta \subset S,\; \Delta \text{ compact }\right\} .
\end{align*}%
\item the short range or Faddeev class (important in scattering theory)%
\begin{equation*}
L_1^{1}\left( S\right) =\left\{ f:\norm{f}_{L_1^1(S)}:=\int_{S}\left( 1+\abs{x}
\right) \abs{f\left( x\right)} dx<\infty \right\} .
\end{equation*}%
\item the Birman-Solomyak spaces ($1\le p<\infty$)
\begin{equation*}
\ell^\infty(L^p(\Reals_\pm)) := \left\{ f: \norm{f}_{\ell^\infty(L^p(\Reals_\pm))} := \sup_{x\in \mathbb{R}_\pm} \left(\int_x^{x\pm1} \abs{f(x)}^pdx \right)^{1/p}<\infty\right\}.
\end{equation*}
\end{itemize}

Next, $\mathfrak{S}_{2}$ denotes the Hilbert-Schmidt class of linear
operators $A$:
\begin{equation*}
\mathfrak{S}_2= \left\{A:\norm{A}_{\mathfrak{S}_{2}}:=%
\left[\operatorname{tr}\left( A^{\ast }A\right)\right]^{1/2}<\infty \right\}
\end{equation*}%
while $\mathfrak{S}_1$ is the trace class
\begin{equation*}
\mathfrak{S}_1= \left\{A:\norm{A}_{\mathfrak{S}_{1}}:=%
\operatorname{tr}\left[\left( A^{\ast }A\right)^{1/2}\right]<\infty \right\}.
\end{equation*}%

$\Spec \left( A\right) $ stands for the spectrum of an operator $A$ and if it is selfadjoint, $\Spec_{ac}\left( A\right)$, $\Spec_d(A)$ denote, respectively, the absolutely continuous and discrete components of $\Spec(A)$.

The following portion of notation will be used extensively in reference to the potential $q$ and quantities associated with it. If $\chi _{S}\left( x\right) $ is the
characteristic function of a set $S$, i.e. $\chi _{S}\left( x\right) =1,x\in
S\,$\ and $0$ otherwise, then we define:
$$q_+(x) := q(x) \chi_{\Reals_+} \;,\; q_-(x)=q(x)\chi_{\Reals_-} \;, \; \tilde{q}(x)=q(x)\chi_{[-a,a]} \text{ for some }a>0. $$

We also denote
$$\delta q := q-\tilde{q} $$
and when the cutoff approximation is taken to infinity, i.e. $a\to\infty$, we write $\delta q \to 0$.

Any quantity $X$ of arbitrary nature (functions,
operators, etc.) related to $\tilde{q}$ will be denoted $\tilde{X}$ and 
\begin{equation*}
\delta X:=X-\widetilde{X}.
\end{equation*}%

\section{Weyl and scattering solutions of the Schr\"odinger equation}

Throughout this section we assume the following. 
\begin{condition}\label{hyp2.1} The potential $q$ is real, locally integrable and such that\footnote{For terminology used without explanation, see e.g. \cite{Deift79,Titch}.}
\begin{enumerate}
\item\label{hyp1} $q$ is limit point case at $-\infty$
\item\label{hyp2} $q_+$ is Faddeev class (short range) and hence 
has the Jost solution at $+\infty$.
\end{enumerate}
\end{condition}

Condition \eqref{hyp1} means that the equation
\begin{equation}\label{star}
-\partial^2_xu+qu=k^2u
\end{equation} 
has a unique, up to a multiplicative constant, solution $\Psi_-$, called Weyl, such that 
$\Psi_-(x,k) \in L^2(\Reals_-)$ for any $k^2\in\C^+$. Condition \eqref{hyp2} implies that $q$ is limit point case at $+\infty$ and the Weyl solution $\Psi_+(x,k)$ can be taken to have the asymptotic behavior:
\begin{equation}\label{psiplusasymp}
\Psi_+(x,k)= e^{ikx}+o(1) \quad,\quad x\to\infty
\end{equation}
for all real\footnote{the Weyl solution $\Psi_+$ coincides in this case with the Jost solution.}
$k$.

In particular, 
we have as in classical scattering theory $\Spec_{ac}(-\partial_x^2+q_+)=\Reals_+$. Furthermore, $\Psi_+,\conj{\Psi}_+$ are both solutions to \eqref{star} for a.e. real $k$ and linearly independent with constant Wronskian
\begin{equation}\label{W1}
W:=W\left(\conj{\Psi_+(x,k)},\Psi_+(x,k)\right)=2ik.
\end{equation} 
So they form a basis of solutions for \eqref{star} for a.e. real $k$ and in particular:
\begin{equation}\label{psiminuspsip}
C(k)\Psi_-(x,k)=\conj{\Psi_+(x,k)}+R(k)\Psi_+(x,k)
\end{equation}
with some $C(k), R(k)$. We call $R$ the reflection coefficient from the right incident. 
Under our hypothesis, neither $C$ nor $R$ can be analytically continued into the upper half plane. Figure \ref{fig1} illustrates a potential from our hypothesis and the asymptotic behavior of $C\Psi_-$ at $+\infty$.
\begin{figure}[htb]
\begin{tikzpicture}[scale=0.6]

\def \qminus {(-8.3,3.5) to[out=40,in=180] (-7.2,3.3) to[out=0,in=180] (-6.2,3.6) .. controls +(1.5,0) and +(-1,0) .. (-3.8,-1) .. controls +(1,0) and +(-1,0) .. (-2.1,3.7) to[out=0,in=165] (0,2.7) }

\def \qplus { (0,2.7) to[out=-15,in=150] (1.5,2.2) to[out=-30,,in=160] (3.3,1.2) to[out=-20,in=135] (5.3,0.6) to[out=-45,in=180] (6.1,-0.3) to[out=0,in=180] (7.4,0.4) to[out=0,in=180] (10,0) }

\def \qplusshaded { \qplus -| (0,2.7)}


\begin{scope}
\clip \qplusshaded;
\foreach \y in {-2,-1.6,...,10}
\draw[domain=0:10,help lines] plot (\x,{8/5*(\x-\y)});
\end{scope}
\draw[help lines] (-8.5,0) -- (10,0)
                  (0,-1) -- +(0,6);
\draw (0,0) node[below left] {$0$};
\draw (-8,-1.5) node[right] {$q_-=q|_{\Reals_-}$ is unknown}; 
\draw (2,-1.5) node[right] {$q_+=q|_{\Reals_+}$ is known};
\draw[thick,blue] 
\qminus \qplus;

\draw[semithick,->] (7,2.3)  -- (4,2.3) arc (270:90:0.5cm) -- (6.5,3.3) node[above] {$e^{-ik x}+R(k)e^{ik x}+o(1)$};

\end{tikzpicture}
\caption{Scattering channels for $q=q_-+q_+$}
\label{fig1}
\end{figure}
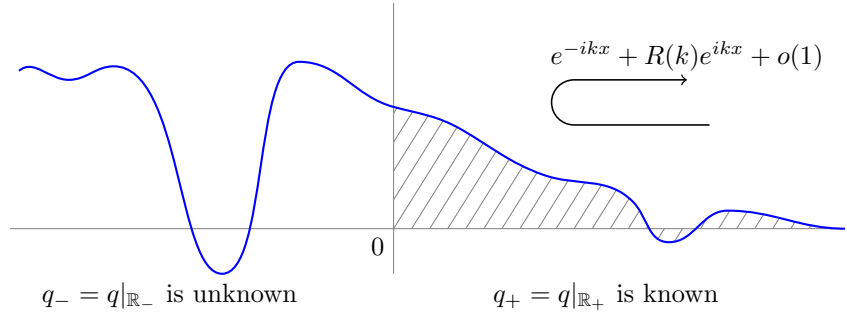

We now consider separately scattering solutions corresponding to $q_\pm$. I.e., first by our hypothesis at $-\infty$, there is a solution $\varphi_-(x,k)$ to
$$-\partial_x^2u+q_-u=k^2u $$
of the form: ($k\in\Reals$)
\begin{equation*}\varphi_-(x,k)=\begin{cases}
D(k)\Psi_-(x,k)\quad&,\quad x<0\\
e^{-ikx}+R_-(k)e^{ikx}\quad&,\quad x\ge0
\end{cases}
\end{equation*}
with some $D(k),R_-(k)$ (see Figure \ref{fig2}). 

\begin{figure}[htb]
\begin{tikzpicture}[scale=0.6]

\def \qminus {(-8.3,3.5) to[out=40,in=180] (-7.2,3.3) to[out=0,in=180] (-6.2,3.6) .. controls +(1.5,0) and +(-1,0) .. (-3.8,-1) .. controls +(1,0) and +(-1,0) .. (-2.1,3.7) to[out=0,in=165] (0,2.7) }

\draw[help lines] (-8.5,0) -- (10,0)
                  (0,-1) -- +(0,6);
\draw (0,0) node[below left] {$0$};
\draw (-8,-1.5) node[right] {$q_-$}; 
\draw[thick,blue] 
\qminus
(0,0)-- +(10,0);

\draw[semithick,->] (7,2.3)  -- (4,2.3) arc (270:90:0.5cm) -- (6.5,3.3) node[above] {$e^{-ik x}+R_-(k)e^{ik x}\phantom{+o(1)}$} ;

\end{tikzpicture}
\caption{Scattering channels for $q_-$}
\label{fig2}
\end{figure}
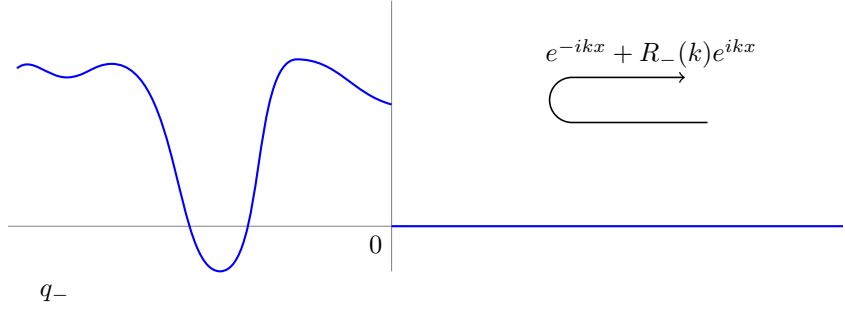

For $q_+$, there exist particular Jost solutions $\varphi_{\ell,+}$ and ${\varphi}_{r,+}$ to
$$ -\partial_x^2u+q_+u=k^2u$$
such that: ($k\in\Reals$)
\begin{align*}
T_+(k){\varphi}_{\ell,+}(x,k)&=\begin{cases}
e^{ikx}+L_+(k)e^{-ikx}\phantom{\dfrac{1}{2}}\qquad\quad\quad&,\quad x<0\\
T_+(k)\Psi_+(x,k)\qquad\quad&,\quad x\ge0
\end{cases} \\
T_+(k)\varphi_{r,+}(x,k)&=\begin{cases}
T_+(k)e^{-ikx}\phantom{\dfrac{1}{2}}\quad&,\quad x<0\\
\conj{\Psi_+(x,k)}+R_+(k)\Psi_+(x,k)\quad&,\quad x\ge0
\end{cases}
\end{align*}
for $k$ real and where $T_+$ is called the transmission coefficient and $L_+,R_+$ the reflection coefficients from the left (right) incident. Because $q_+$ is short range, $T_+$ can be analytically continued in the upper half plane and has only a finite number of poles $\{i\varkappa_n^+\}_{n=1}^N$. Since, in addition, $q_+$ is supported on the right half line, $L_+$ can also be analytically extended to the upper half plane and shares the same poles as $T_+$. However, in general, $R_+$ can not be extended off the real axis. The asymptotic behavior of these solutions is illustrated in Figure \ref{fig3}.

\begin{figure}[htb]
\begin{tikzpicture}[scale=0.6]

\def \qminus {(-8.3,3.5) to[out=40,in=180] (-7.2,3.3) to[out=0,in=180] (-6.2,3.6) .. controls +(1.5,0) and +(-1,0) .. (-3.8,-1) .. controls +(1,0) and +(-1,0) .. (-2.1,3.7) to[out=0,in=165] (0,2.7) }

\def \qplus { (0,2.7) to[out=-15,in=150] (1.5,2.2) to[out=-30,,in=160] (3.3,1.2) to[out=-20,in=135] (5.3,0.6) to[out=-45,in=180] (6.1,-0.3) to[out=0,in=180] (7.4,0.4) to[out=0,in=180] (10,0) }

\def \qplusshaded { \qplus -| (0,2.7)}

\begin{scope}
\clip \qplusshaded;
\foreach \y in {-2,-1.6,...,10}
\draw[domain=0:10,help lines] plot (\x,{8/5*(\x-\y)});
\end{scope}
\draw[help lines] (-8.5,0) -- (10,0)
                  (0,-1) -- +(0,6);
\draw (0,0) node[below left] {$0$};
\draw (2,-1.5) node[right] {$q_+$};
\draw[thick,blue] 
 \qplus
 (-8.5,0)-- (0,0);

\draw[semithick,->] (9.5,2.3)  -- (6.5,2.3) arc (270:90:0.5cm) -- (9,3.3) node[above left] {$e^{-ik x}+R_+(k)e^{ik x}+o(1)$};
\draw[semithick,->] (6.5,2.3) -- +(-1.5,0) node[above left] {$T_+(k)e^{-ikx}$} ;

\draw[semithick,->] (-8,2.3)  -- (-5,2.3) arc (-90:90:0.5cm) -- (-7.5,3.3) node[above] {$e^{ik x}+L_+(k)e^{-ik x}$} ;
\draw[semithick,->] (-5,2.3) -- +(1.5,0) node[above right] {$T_+(k)e^{ikx}$} ;
\draw (-2,2.2) node[right] {$+o(1)$};

\end{tikzpicture}
\caption{Scattering channels for $q_+$}
\label{fig3}
\end{figure}
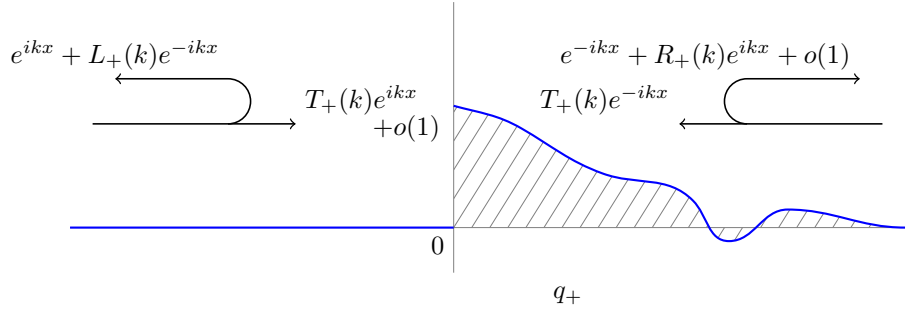

Note that all right reflection and transmission coefficients can be expressed in terms of Wronskians. Of particular interest, we have for a.e. real $k$:
\begin{align*}
W&=2ik=W (\Psi_+(x,k),C(k)\Psi_-(x,k))  \\
R(k) &= \frac{W\left(\conj{\Psi_+(x,k)},C(k)\Psi_-(x,k)\right)}{W} \\
R_-(k) &= \left. \frac{W\left(e^{-ikx},\varphi_-(x,k)\right)}{W} \right|_{x \ge0}\\ 
R_+(k) &=\left.\frac{W\left(\conj{\Psi_+(x,k)},T_+(k)\varphi_{r,+}(x,k)\right)}{W} \right|_{x\ge0}\\
\frac{1}{T_+(k)}&= \frac{W\left(\varphi_{r,+}(x,k),{\varphi}_{\ell,+}(x,k)\right)}{W}
\end{align*}

Note that any truncation $\tilde{q}=q|_{[-a,a]}$ 
is compactly supported which implies that $\widetilde{R},\widetilde{R}_+$ can be analytically continued into $\C^+$ \cite{Deift79} except at a finite number of poles. These poles are located on $i\Reals_+$ such that their squares correspond respectively to the discrete spectrum of $-\partial_x^2+\tilde{q}$ and $-\partial_x^2+\tilde{q}_+$. 
 
\section{The Titchmarsh-Weyl $m$-function}

In this section, we review properties of the Titchmarsh-Weyl $m$-function which is the logarithmic derivative of the Weyl solutions $\Psi_\pm(x,k)$ as $x\to\pm0$. It will be a central object in redefining scattering quantities in the next section. 
We will have to impose some additional conditions on the potential $q$. Most of the material already appeared in \cite{RybIP09} but are repeated here for the reader's convenience.

\begin{definition} The Titchmarsh-Weyl $m$-function is defined by:
$$m_{\pm}(k^2)= \left.\pm \frac{\partial_x \Psi_{\pm}(x,k)}{\Psi_{\pm}(x,k)} \right|_{x=\pm0}$$
\end{definition}

Some of the important properties of the $m$-function are (see e.g. \cite{Ramm,RybIP09}):
\begin{itemize}
\item $m_\pm$ is analytic for all $k^2\in\C^+$ and has the Herglotz property, i.e. $m_\pm:\C^+ \to \C^+$
\item symmetry $m_\pm(\conj{z})=\conj{m_\pm(z)}$
\item the singularities of $m_\pm$ correspond to the spectrum of the half line Dirichlet Schr\"odinger operator, i.e. $-\partial_x^2+q$ on $\Reals_\pm$ with $u(\pm0)=0$.
\item the Borg-Marchenko uniqueness theorem: $m_1=m_2 \;\Rightarrow\; q_1=q_2$. 
\end{itemize}

The following representation of $m_\pm$ will be useful.

\begin{proposition}\label{pr3.1}
Let $q$ be a real function on $\Reals$  
such that $q \in \ell^\infty(L^2(\Reals_-))\cap L^1(\Reals_+)$. 
Let $\gamma=\max(\gamma_-,\gamma_+)$ where 
\begin{align*}
\gamma_- &= \max \left( \sqrt{2\norm{q_-}_{\ell^\infty(L^2(\Reals_-))}}, e\norm{q_-}_{\ell^\infty(L^2(\Reals_-))}\right),\\
\gamma_+ &= \frac{\norm{q_+}_{L^1(\Reals_+)}}{2}.
\end{align*}
Then for $k=\alpha+ih$, $h>\gamma$,
\begin{equation}\label{mA}
m_\pm(k^2)=ik\mp\int^{\pm\infty}_0e^{\pm2ikx}A_\pm(x)dx
\end{equation}
with some real function $A_\pm(x)$, called the $A$-amplitude. 
The integral in \eqref{mA} is absolutely convergent and the $A$-amplitude has the following properties.
\begin{enumerate}
\item $A_\pm-q_\pm$ is continuous on $\Reals_\pm$ and for $\pm x > 0$:
\begin{equation}\label{Amq}
\abs{A_\pm(x)-q_\pm(x)} \le \left( \pm\int^x_0 \abs{q_\pm(s)}ds\right)^2e^{\pm2\gamma x}
\end{equation}
\item\label{it1} If $q_1,q_2 \in \ell^\infty(L^2(\Reals_-))\cap L^1(\Reals_+)$ then 
\begin{equation}\label{q1q2}
q_1(x)=q_2(x) \text{ on }[0,\pm a] 
\quad \Rightarrow\quad A_1(x)=A_2(x)\text{ on }[0,\pm a]. 
\end{equation}
\item\label{it3} For any $h>\gamma$,
$$\norm{e^{\mp2hx}A_\pm(x)}_{L^1(\Reals_\pm)} \le C(h,q_\pm) <\infty$$
and $C(h,q_\pm)$ is a nonincreasing function of $h$.
\item\label{it4} For any $h>\gamma$,
$$e^{2hx}A_-(x) \in L^2(\Reals_-).$$
\end{enumerate}
\end{proposition}

\begin{proof} The representation \eqref{mA} appeared in \cite{Ramm} for short range $q$'s and in \cite{Simon99} for more general $q$'s. 
Properties \eqref{Amq}-\eqref{q1q2} were derived for $q_+ \in L^1(\Reals_+)$ in \cite{Simon99} then for $q_+ \in \ell^\infty(L^1(\Reals_+))$ in \cite{AMR07} but since
$$m_-(q_-(x),k^2)=m_+(q_-(-x),k^2) $$ 
and $\ell^\infty(L^2(\Reals_-))\subset\ell^\infty(L^1(\Reals_-))$, we have adjusted the results accordingly. So only \eqref{it3}-\eqref{it4} require a proof. We will consider $A_-$ and $p=1,2$. Using Minkowski's inequality, one needs only to show $e^{2hx} q_-(x)$ and $e^{2hx}(A_-(x)-q_-(x))$ are in $L^p(\Reals_-)$. Dropping the subscripts, we have
\begin{align*}
\int_{-\infty}^0 \abs{e^{2hx}q(x)}^pdx &= \int_{-\infty}^0e^{2hpx}\abs{q(x)}^pdx \\
&= \sum_{n=0}^{\infty} \int_{-n-1}^{-n}e^{2hpx}\abs{q(x)}^pdx \\
& \le \sum_{n=0}^\infty e^{-2hpn} \int_{-n-1}^{-n}\abs{q(x)}^pdx \\
& \le \sum_{n=0}^\infty e^{-2hpn} \norm{q(x)}^p_{\ell^\infty(L^p(\Reals_-))} \\
&= \frac{1}{1-e^{-2hp}}\norm{q(x)}^p_{\ell^\infty(L^p(\Reals_-))}.
\end{align*}
For the next term, we will make use of the following: ($x \le 0$)
\begin{align*}
\int_x^0\abs{q(s)}ds & \le \sum_{n=1}^{-\left\lfloor x\right\rfloor}\int_{-n}^{-n+1} \abs{q(s)}ds \le \sum_{n=1}^{-\left\lfloor x\right\rfloor}\norm{q}_{\ell^\infty(L^1(\Reals_-))}\\
&\le (1-x)\norm{q}_{\ell^\infty(L^1(\Reals_-))} \le (1-x)\norm{q}_{\ell^\infty(L^p(\Reals_-))}.
\end{align*}
where the last inequality is a direct consequence of H\"{o}lder's inequality. So,
\begin{align*}
\int_{-\infty}^0\abs{e^{2hx}(A(x)-q(x))}^pdx &= \int_{-\infty}^0e^{2hpx}\abs{A(x)-q(x)}^pdx \\
&\le \int_{-\infty}^0e^{2p(h-\gamma)x} \left(\int_x^0\abs{q(s)}ds\right)^{2p}dx \\
&\le \int_{-\infty}^0e^{2p(h-\gamma)x}(1-x)^{2p}\norm{q}_{\ell^\infty(L^p(\Reals_-))}^{2p}dx.
\end{align*}
One readily verifies that for any $m=0,1,2,\ldots$ and $b>0$
$$\int_{-\infty}^0(1-x)^me^{bx}dx= \frac{m!}{b^{m+1}}\sum_{k=0}^m \frac{b^k}{k!}. $$  
Therefore, we have
\begin{align*}
\int_{-\infty}^0\abs{e^{2hx}(A(x)-q(x))}^pdx &\le \frac{2p!}{\left[2p(h-\gamma)\right]^{2p+1}}\sum_{j=0}^{2p} \frac{\left[2p(h-\gamma)\right]^{j}}{j!}\norm{q}_{\ell^\infty(L^p(\Reals_-))}^{2p}. 
\end{align*}
So \eqref{it3}-\eqref{it4} are verified with
$$C(h,q_-)= \frac{1}{1-e^{-2h}}\norm{q_-(x)}_{\ell^\infty(L^1(\Reals^-))}+\frac{1}{4\left(h-\gamma\right)^{3}}\sum_{j=0}^{2} \frac{\left[2(h-\gamma)\right]^{j}}{j!}\norm{q_-}_{\ell^\infty(L^1(\Reals_-))}^{2}.$$
Similarly for $A_+$, we can take in \eqref{it3} 
\begin{equation*}C(h,q_+)=\norm{q_+}_{L^1(\Reals_+)}+\frac{1}{h-\gamma}\norm{q_+}_{L^1(\Reals_+)}^2 . \qedhere
\end{equation*}
\end{proof}

\begin{remark}\label{rem3.3} In the case of a truncated potential $\tilde{q}$, since $\widetilde{\gamma}_{\pm}\le \gamma_{\pm}$ and $C(h,\tilde{q}_{\pm})\le C(h,q_\pm)$, all above results remain true for the same $h$. If, in addition, $q_+\in L^2_{loc}(\Reals_+)$ then $\tilde{q}_+ \in \ell^\infty(L^2(\Reals_+))$ and thus $e^{-2hx}\widetilde{A}_+(x) \in L^2(\Reals_+)$ for $h$ large enough.
\end{remark}

\begin{corollary}\label{cor3.4} Let $q \in \ell^\infty(L^2(\Reals_-))$, and let $h \ge \gamma_-$ where $\gamma_-$ is defined as in Proposition \ref{pr3.1}. Then
\begin{enumerate}
\item $ik-m_{-}(k^2) \in L^2(\Reals+ih)$,
\item if $\tilde{q}_- = q_-|_{[-a,0]}$ for some $a>0$, then
$$\delta m_-(k^2):= m_-(k^2)-\widetilde{m}_-(k^2) \to 0 \text{ in } L^2(\Reals+ih)\quad,\quad a \to\infty. $$
\end{enumerate}
\end{corollary}

\begin{proof} Note that for $k=\alpha+ih$ where $\alpha\in\Reals, h>\gamma_-$,
$$ik-m_{-}(k^2)= \int_{-\infty}^0 e^{-2i\alpha x}e^{2hx}A_-(x)dx $$
where by Proposition \ref{pr3.1}, $e^{2hx}A_-(x)\in L^2(\Reals+ih)$. The Plancherel formula in our case takes the form:
$$\norm{ \int_{\Reals_-} e^{-2i\alpha x} f(x)dx}_{L^2(\Reals)}= \sqrt{\pi} \norm{f(x)}_{L^2(\Reals_-)} $$
and hence
\begin{align*}
\norm{ik-m_-(k^2)}_{L^2(\Reals+ih)}&= \sqrt{\pi}\norm{e^{2hx}A_-(x)}_{L^2(\Reals_-)}<\infty  
\end{align*}

In the case $\tilde{q}_-= q_-|_{[-a,0]}$, we have $\widetilde{\gamma}_-\le \gamma_-$ and by Proposition \ref{pr3.1} \eqref{it1}, 
$$\delta m_-(k^2)= \int_{-\infty}^{-a}e^{-2ik x}\delta A(x)dx$$
and therefore for any $h>\gamma_-$,
\begin{equation*}\norm{\delta m_-(k^2)}_{L^2(\Reals+ih)}= \sqrt{\pi}\norm{e^{2hx}\delta A(x)}_{L^2((-\infty,-a])} \to 0 \quad,\quad a \to\infty. \qedhere
\end{equation*}
\end{proof}

\section{The reflection and transmission coefficients}

In this section we establish some properties of one of our main objects:
\begin{equation*}
G(k):= \Delta R(k)=R(k)-R_+(k).
\end{equation*}

As mentioned when first introduced, neither $R$ nor $R_+$ can be analytically extended to the upper half plane for a potential $q$ under the conditions of Hypothesis \ref{hyp2.1} (or those in Proposition \ref{pr3.1}). But by rewriting $G$ exclusively in terms of $R_-$, $T_+$, and $L_+$, we will see that $G$ can be analytically extended to the upper half plane. We also derive key properties of $R_-,T_+,L_+$ and $G$ in $\C^+$ which will be used later to recover $q_-(x)$ -- assuming $R,R_+$ are known.

First we rewrite the reflection and transmission coefficients in terms of the $m$-function. Setting 
$$m_\pm=m_\pm(k^2+i0):=\lim_{\varepsilon\to0^+}m_\pm(k^2+i\epsilon),$$
we have:
\begin{align}
R(k)\phantom{+}   &= -\frac{m_-+\conj{m}_+}{m_-+m_+}\frac{\conj{\Psi_+(0,k)}}{\Psi_+(0,k)}, \label{R} \\
R_+(k) &=-\frac{\conj{m}_++ik}{m_++ik}\frac{\conj{\Psi_+(0,k)}}{\Psi_+(0,k)}, \label{Rp} \\
R_-(k) &= \frac{ik-m_-}{ik+m_-}, \label{Rm}\\
L_+(k) &=\frac{ik-m_+}{ik+m_+}, \label{Lp}\\
T_+(k) &=\frac{2ik}{(ik+m_+)\Psi_+(0,k)} \label{Tp}.
\end{align}

The above are obtained using continuity of the various solutions and their derivatives in $x$ at the point $x=0$ or, alternately, the Wronskians. Recall that the above are defined for a.e. real $k$. 

From \eqref{Rm}-\eqref{Lp} and properties of the Titchmarsh-Weyl $m$-function, one notes that $L_+,R_-$ have a meromorphic extension to the upper half plane. By Proposition \ref{pr3.1} \eqref{it3}, $L_+,R_-$ are smooth on $\Reals+ih$ for any $h>\gamma$. Furthermore, by the Borg-Marchenko uniqueness theorem $L_+,R_-$ determine uniquely respectively $q_\pm(x)$ \cite{Ramm}. 

Now, using \eqref{W1} and \eqref{R}-\eqref{Tp}, one readily verifies that:
\begin{equation}\label{G}
G(k)= \frac{T_+^2(k)R_-(k)}{1-L_+(k)R_-(k)}
\end{equation}
and thus $G$ can be analytically extended to the upper half plane (recall that $T_+$ also has an analytic extension since $q_+$ is short range \cite{Deift79}).

\begin{proposition}\label{pr4.1}Let $q$ be as in Proposition \ref{pr3.1} and let 
$$ h_\pm= \inf \left\{h\; : \; h>\gamma_\pm \text{ and } C(h,q_{\pm})<\frac{h}{2}\right\}$$
Then for all $h>\max(h_+,h_-)$ 
\begin{enumerate}
\item $R_-,L_+ \in \left(L^\infty\cap L^2\right)(\Reals+ih)$ with their $L^\infty$ norm no greater than $1/3$.
\item $\delta L_+ \; \to \; 0$ in $L^\infty(\Reals+ih)$ when $a\to\infty$. 
\item $R_- \in L^1(\Reals+ih)$.
\item $\delta R_- \; \to \; 0$ in $L^1(\Reals+ih)$.
\end{enumerate}
\end{proposition}

\begin{proof} The above results are direct consequences of Proposition \ref{pr3.1}. 
For all $k\in\Reals+ih$, such that $h>\max(h_+,h_-)$,
\begin{align*}
\abs{ik-m_\pm(k^2)} & = \abs{\int_0^{\pm\infty}e^{\pm 2ikx}A_\pm(x)dx}\\
& \le \norm{e^{\mp2hx}A_\pm(x)}_{L^1(\Reals_\pm)} \le \frac{h}{2} \\
\abs{ik+m_\pm(k^2)} &= \abs{2ik\mp\int_0^{\pm\infty}e^{2ikx}A_\pm(x)dx} \\
&\ge2\abs{k} \cdot\abs{1 - \frac{1}{2\abs{k}}\abs{ \int_0^{\pm\infty}e^{\pm2ikx}A_{\pm}(x)dx}\;} \\
&\ge \frac{3\abs{k}}{2} 
\end{align*}
where we have used $\abs{k} \ge h$. Thus $\frac{1}{ik+m_\pm(k^2)}\in L^2(\Reals+ih)$ and it follows immediately from \eqref{Rm}-\eqref{Lp} that $\norm{R_-(k)}_{L^\infty(\Reals+ih)}, \norm{L_+(k)}_{L^\infty(\Reals+ih)} \le \frac{1}{3}$ and $L_+,R_-\in L^2(\Reals+ih)$.

We further obtain $R_-\in L^1(\Reals+ih)$ by the Cauchy-Schwartz inequality for $h>h_-$:
\begin{equation*}
\norm{R_-(k)}_{L^1(\Reals+ih)} \le \norm{\frac{1}{ik+m_-(k^2)}}_{L^2(\Reals+ih)}\cdot\norm{ik-m_-(k^2)}_{L^2(\Reals+ih)}. 
\end{equation*}
By Remark \ref{rem3.3}, the above is also true for $\widetilde{R}_-,\widetilde{L}_+$.

Now for all $k\in\Reals+ih$ with $h>h_+$,
\begin{align*}
\abs{\delta L_+(k)} &= \abs{ \frac{-2ik \delta m_+(k^2)}{(ik+m_+(k^2))(ik+\widetilde{m}_+(k^2))}} \leq \frac{8}{9h} \abs{\delta m_+(k^2)} \\
& \le \frac{8}{9h} \norm{e^{-2hx}\delta A_+(x)}_{L^1([a,\infty))} \; \to \; 0 \quad,\quad a \to \infty.
\end{align*}

We also have
$$\delta R_-(k)= \frac{-2ik \delta m_-(k^2)}{(ik+m_-(k^2))(ik+\widetilde{m}_-(k^2))}.$$ 

Using $L^\infty$ norms and the Cauchy-Schwartz inequality:
\begin{align*}
\norm{\delta R_-}_{L^1(\Reals+ih)} &\le \norm{\frac{-2ik}{ik+\widetilde{m}_-(k^2)}}_{L^\infty(\Reals+ih)}\cdot \norm{\frac{\delta m_-(k^2)}{ik+{m}_-(k^2)}}_{L^1(\Reals+ih)}\\
& \le \frac{4}{3}\norm{\frac{1}{ik+{m}_-(k^2)}}_{L^2(\Reals+ih)}\cdot\norm{\delta m_-(k^2)}_{L^2(\Reals+ih)} 
\end{align*}
and the right hand side, by Corollary \ref{cor3.4}, goes to zero when $a\to\infty$.
\end{proof}

\begin{remark} Property \eqref{it3} will play a crucial role. Note that if $q(x)=c\delta(x)$, where $\delta$ is Dirac's $\delta$-function, then  
\begin{equation*}
R(k)= \frac{c}{2ik-c}
\end{equation*}
which is not in $L^1(\Reals+ih)$. This suggests that the condition $\ell^\infty(L^2(\Reals))$ may not be relaxed to read $\ell^\infty(L^1(\Reals))$. 
\end{remark}

\begin{corollary}\label{cor5.3} For any finite $z$, and $q_-,h$ under the conditions of Proposition \ref{pr4.1}, $e^{2ikz}R_-(k) \in L^1(\Reals+ih)$.
\end{corollary}
\begin{proof} Immediately follows from
$ \norm{e^{2ikz}R_-(k)}_{L^1(\Reals+ih)}= e^{-2hz}\norm{R_-(k)}_{L^1(\Reals+ih)}$.
\end{proof}

While trivial, the above corollary plays an important part in our arguments. The reflection coefficient for the shifted 
potential $
q(x+z)$ 
is $R(k)e^{2ikz}$ where $R(k)$ is the reflection coefficient corresponding to $
q(x)$. 

\begin{lemma}\label{lem4.3} Let $q_+ \in L^1_1(\Reals_+)$ and let $h>\beta$ where
$$\beta =2 \max\left\{\norm{q_+}_{L^1_1(\Reals_+)},1\right\}. $$
Then 
$$\norm{T_+(k)}_{L^\infty(\Reals+ih)},\norm{\widetilde{T}_+(k)}_{L^\infty(\Reals+ih)} \le 2^\beta$$
and $\delta T_+(k) \to 0$ in $L^\infty(\Reals+ih)$ as $a\to\infty$.
\end{lemma}

\begin{proof} The following are well-known facts (e.g. \cite{Deift79}) for $q_+$ short range and supported on $\Reals_+$:
\begin{enumerate}
\item $T_+(k)$ is analytic in $\C^+$ except at a finite number of simple poles $\{i\varkappa_n^+\}_{n=1}^N$ where
\begin{equation}\label{N}
N \le 1+ \int_{\Reals_+}\abs{x}\abs{q_+(x)}dx.
\end{equation}
\item $\abs{T_+(k)}^2+\abs{L_+(k)}^2=1$ for a.e. real $k$ and $T_+(-k)=\conj{T_+(k)}$. 
\item $T_+(k)$ admits the following representation for any $k \in \C^+$:
\begin{equation*}
T_+(k)= \prod_{n=1}^N \frac{k+i\varkappa_n^+}{k-i\varkappa_n^+}\exp\left({\frac{i}{\pi}\int_\Reals \frac{\log \abs{T_+(\omega)}^{-1}}{\omega-k}d\omega}\right).
\end{equation*}
\end{enumerate}
We also have the Lieb-Thirring inequality \cite{Weidl}
\begin{equation}\label{LT}
\sum_{n=1}^N \varkappa_n^+ \le L_{1/2,1} \int_{\Reals_+} \abs{q_+(x)}dx
\end{equation}
where $1/2 \le L_{1/2,1} \le 1.005$. Thus, for any $k \in \Reals+ih$, $h>\beta$
\begin{align*}
\abs{T_+(k)} &= \abs{\prod_{n=1}^N\frac{k+i\varkappa_n^+}{k-i\varkappa_n^+}} \exp\left({\frac{1}{\pi} \int_\Reals \Re \frac{i}{\omega-k} \log{\abs{T_+(\omega)}^{-1}}d\omega}\right)\\
&= \prod_{n=1}^N \sqrt{1+\frac{4h\varkappa_n^+}{(h-\varkappa_n^+)^2}}\exp\left({\frac{-h}{\pi}\int_\Reals \frac{\log\abs{T_+(\omega)}^{-1}}{h^2+(\omega-\alpha)^2}d\omega}\right).
\end{align*}
Since $\abs{T_+(\omega)} \le 1$ for a.e. real $\omega$, the above becomes
$$\abs{T_+(k)} \le \prod_{n=1}^N \sqrt{1+\frac{4h\varkappa_n^+}{(h-\varkappa_n^+)^2}}. $$
But by \eqref{LT}, we also have for each $n$:
\begin{equation}\label{kappaLT}
\varkappa_n^+ \le L_{1/2,1} \norm{q_+}_{L^1(\Reals_+)} \le \frac{2}{3}\beta < \frac{2}{3}h
\end{equation}
and hence for all $k \in \Reals+ih$ where $h>\beta$
\begin{equation}\label{Tpinf}
\abs{T_+(k)} \le \left( \frac{11}{3} \right)^{\beta/2} \le 2^\beta
\end{equation}
where we have used $N \le \beta$ from \eqref{N}. Inequalities \eqref{N}-\eqref{Tpinf} are also valid for $\widetilde{T}_+(k)$ and thus $T_+,\widetilde{T}_+$ are uniformly bounded on $\Reals+ih$.

From \cite{Deift79}, we now use the following results 
\begin{align}
\dfrac{1}{T_+(k)} &= 1-\dfrac{1}{2ik} \int_{\Reals} q_+(x)y_+(x,k)dx \notag \\
\text{where } y_+(x,k)&:=e^{-ikx} \varphi_{\ell,+}(x,k) \text{ satisfies for }  \Im k \ge 0, \notag \\
y_+(x,k) &= 1+ \int_x^\infty D_k(t-x)q_+(t)y_+(t,k)dt \;,\; D_k(y)= \frac{e^{2iky}-1}{2ik}, \label{yp}\\
\abs{y_+(x,k)} &\le K(\beta) (1+\abs{x}) \notag 
\end{align}
where $K$ is a constant depending only on $\beta$. Note that 
\begin{equation}\label{dT}
\delta{T_+(k)}= \frac{\tilde{T_+(k)T_+(k)}}{2ik}\left[ \int_{\Reals}q_+(x)\delta y_+(x,k)dx+\int_\Reals\tilde{y}_+(x,k)\delta q_+(x)dx\right].
\end{equation}
From \eqref{yp},
\begin{equation}\label{eq5.12'}
\delta y_+(x,k)= \int_x^\infty D_k(t-x)\tilde{y}_+(t,k)\delta q_+(t)dt+\int_x^\infty D_k(t-x)q_+(t)\delta y_+(t,k)dt
\end{equation}
and since for $k\ne0$, $\abs{D_k(y)} \le \frac{1}{\abs{k}}$ for all $y \ge 0$, $\Im k \ge 0$,
\begin{align*}
\abs{\delta y_+(x,k)} &\le \frac{K(\beta)}{\abs{k}}\norm{\delta q_+}_{L^1_1(\Reals_+)}+ \int_x^\infty \frac{\abs{q_+(t)}}{\abs{k}}\abs{\delta y_+(t,k)}dt \\
&\le \frac{K(\beta)}{\abs{k}}e^{\frac{\beta}{2\abs{k}}}\norm{\delta q_+}_{L^1_1(\Reals_+)} \quad,\quad k \ne 0 \quad,\quad \Im k \ge0
\end{align*}
by iteration on the Volterra integral equation for $\frac{\delta y_+(x,k)k}{K(\beta)\norm{\delta q_+}_{L_1^1(\Reals)}}$ derived from \eqref{eq5.12'}. Hence for \eqref{dT}, we have
\begin{equation*}
\norm{\delta T_+(k)}_{L^{\infty}(\Reals+ih)} \le 2^{2\beta}\frac{K(\beta)}{\beta} \norm{\delta q_+}_{L^1_1(\Reals_+)}
\end{equation*}
where the right hand side goes to zero for $a\to\infty$ by the dominated convergence theorem.
\end{proof}

\begin{proposition}\label{pr5.5} Under Hypothesis \ref{hyp2.1},
$$\Delta R(k) := R(k)-R_+(k) $$
is analytic in $\C^+$ except on a set
$$\mathcal{S}= \left\{ i\varkappa_n^+\right\} \cup \sigma \subset i\Reals_+$$
where 
$$\left\{ -\left(\varkappa_n^+\right)^2\right\}=\Spec_d(-\partial_x^2+q_+) \quad , \quad \sigma=\left\{ \lambda: \lambda^2 \in \Spec(-\partial_x^2+q)\cap\Reals_-\right\} . $$
\end{proposition}

\begin{proof} 
By direct computation, we have 
\begin{equation} \label{G5.5}
G(k) 
=T_+(k)
\frac{ik-m_-(k^2)}{m_+(k^2)+m_-(k^2)}g(k)
\end{equation}
where
$$g(k)= \frac{T_+(k)}{1+L_+(k)}=\frac{1}{\Psi_+(0,k)}.$$

We gather the following facts:
\begin{enumerate}
\item it is well-known that $T_+(k),L_+(k)$ are analytic in $\C^+\setminus\left\{i\varkappa_n^+\right\}_{n=1}^{N}$ where $\left\{-\left(\varkappa_n^+\right)^2\right\}$ is the negative simple discrete spectrum of $-\partial_x^2+q_+$. So we also have that $g(k)$ is meromorphic in $\C^+$ but with poles different from those of $T_+,L_+$. Poles of $g(k)$ correspond to the poles of $m_+$, i.e. $\kappa$'s such that $\Psi_+(0,\kappa)=0$. 
\item recall that $m_\pm$ is analytic in $\C^+$ except for some singularities\footnote{Singularities of $m_+$ are a finite number poles since $q_+$ is short range \cite{Ramm} whereas the set of singularities of $m_-$, while bounded by $\gamma_-$, need not be made of poles -- it could be continuous.} $\kappa\in i\Reals_+$, hence so is $\frac{ik-m_-(k^2)}{(m_-+m_+)(k^2)}$. Note that $m_+(k^2)+m_-(k^2)=0$ corresponds to $W(\psi_-,\psi_+)=0$. But then if these two solutions to $\partial_x^2+q$ are linearly dependent, then $\psi_\pm(x,k)\in L^2(\Reals)$ and so 
$\lambda\in\sigma$. 
\item we now consider singularities of $m_\pm$. Those cases correspond exactly to $\psi_{\pm}(0,\kappa)=0$ with $\psi_{\pm}'(0,\kappa)\ne0$. But if $\psi_-(0,\kappa)=0$, we can assume $\psi_+(0,\kappa)\ne0$ (otherwise $\kappa\in\sigma$) and so $m_+(\kappa^2)$ finite. Therefore, by \eqref{G5.5},
$$G(\kappa)=-\frac{T(\kappa)}{\psi_+'(0,\kappa)} $$
is finite unless $\kappa\in\{i\varkappa_n^+\}$. Now if $\psi_+(0,\kappa)=0$, we have by \eqref{Tp} that ($\kappa\ne0$)
$$T_+(\kappa)=\frac{2i\kappa}{\psi_+'(0,\kappa)} $$
is finite and by \eqref{G5.5} since we can assume $m_-(\kappa^2)$ is finite, then
$$G(\kappa)= T_+(\kappa)\frac{i\kappa+m_-(\kappa^2)}{\psi_+'(0,\kappa)} $$
is finite too.
\end{enumerate}
Thus, we find that $G(k)$ is analytic in $\C^+\setminus \left(\left\{ i\varkappa_n^+\right\} \cup \sigma\right)$.
\end{proof}

\begin{remark} By Proposition \ref{pr3.1}, we have $\gamma > \sup\abs{\mathcal{S}}$ and thus $G(k)$ is smooth  
on $\Reals+ih$ for any $h>\gamma$. 
Note that $\sigma$ need not be finite, but is bounded.
\end{remark}

\begin{proposition}\label{pr4.3} Let $q$ be a real function on $\Reals$ such that $q\in \ell^\infty(L^2(\Reals_-))\cap L^1_1(\Reals_+)$ and let $h>h_0$ where $h_0=\max(h_+,h_-,\beta)$ where $h_\pm$ are as in Proposition \ref{pr4.1} and $\beta$ as in Lemma \ref{lem4.3}. Then $$G(k)=\Delta R(k) \in L^1(\Reals+ih)\quad\text{ and }\quad\delta G(k) \to 0\text{ in }L^1(\Reals+ih).$$ 
\end{proposition}

\begin{proof} Omitting the variable $k$ for brevity, from \eqref{G},  
Proposition \ref{pr4.1} and Lemma \ref{lem4.3}, we have:
$$ \norm{G}_{L^1(\Reals+ih)} \le \frac{9}{8}\cdot 2^{2h_0} \cdot 
\norm{R_-}_{L^1(\Reals+ih)} <\infty. $$

By direct computation,
\begin{equation}\label{dG}
\delta G = \frac{G_1 \delta L_++G_2\delta T_++G_3\delta R_-}{(1-L_+R_-)(1-\widetilde{L}_+\widetilde{R}_-)}
\end{equation}
where 
$$G_1= {{T}_+^2R_-\widetilde{R}_-} \quad,\quad G_2={R_-}(1-{L}_+\widetilde{R}_-)(T_++\widetilde{T}_+) \quad,\quad G_3={\widetilde{T}_+^2}. $$
By Proposition \ref{pr4.1},
$$\norm{\frac{1}{(1-L_+R_-)(1-\widetilde{L}_+\widetilde{R}_-)}}_{L^\infty(\Reals+ih)} \le \left(\frac{9}{8}\right)^2 $$
so it is enough to show that each term inside the brackets in \eqref{dG} goes to zero in $L^1(\Reals+ih)$ as $a \to\infty$. Indeed
\begin{align*}
\norm{G_1\delta L_+}_{L^1(\Reals+ih)} &\le \frac{2^{2h_0}}{3}
\cdot\norm{R_-}_{L^1(\Reals+ih)}\cdot\norm{\delta L_+}_{L^\infty(\Reals+ih)} \; \to \; 0, \\
\norm{G_2\delta T_+}_{L^1(\Reals+ih)} &\le \frac{10}{9}\cdot{2^{2h_0}}\norm{R_-}_{L^1(\Reals+ih)}
\cdot\norm{\delta T_+}_{L^\infty(\Reals+ih)} \; \to \; 0, \\
\norm{G_3\delta R_-}_{L^1(\Reals+ih)} &\le 2^{2h_0}
\cdot\norm{\delta R_-}_{L^1(\Reals+ih)} \; \to \; 0. \qedhere
\end{align*}
\end{proof}

\begin{corollary}\label{cor5.8}Let $z$ be a fixed real parameter and let $q,h$ be as in Proposition \ref{pr4.3}. Then $G_z(k):=e^{2ikz}\Delta R(k) \in L^1(\Reals+ih)$ and $\delta G_z(k) \to 0$ in $L^1(\Reals+ih)$.
\end{corollary}
\begin{proof} Note that $G_z(k)$ correspond to the shifted potential $q(x+z)$. For such potential, $R_-(k)$ becomes $R_-(k)e^{2ikz}$, $T_+(k)$ remains the same and $L_+(k)$ becomes $L(k)e^{-2ikz}$. So by Corollary \ref{cor5.3} and Proposition \ref{pr4.3}, $G_z(k) \in L^1(\Reals+ih)$ and $\delta G_z(k)\to 0$ in $L^1(\Reals+ih)$.
\end{proof}

\section{A trace class operator}

In this section we introduce a lemma which appeared in a more general form in \cite{Ryprep} and will be a central argument in the main result of this paper in the next section.

\begin{proposition}\label{pr6.2}
Let $A$ be smooth on $\Reals+ih$ for some $h>0$, and $A \in L^1(\Reals+ih)$. Then the integral operator $\mathbb{A}$ on $L^2(\Reals_+)$ with kernel
$$\mathbf{A}(x,y)=
\int_{\Reals+ih}e^{ik(x+y)}A(k)\frac{dk}{2\pi} 
\quad,\quad x,y \ge0$$
is trace class, and
$$\norm{\mathbb{A}}_{\mathfrak{S}_1} \le \frac{1}{4\pi h}\norm{A}_{L^1(\Reals+ih)}. $$
\end{proposition}

\begin{proof}
Denote $A_h(\alpha)=A(\alpha+ih)$ and $\widehat{f}(z)=\frac{1}{2\pi}\int_\Reals e^{ikz}f(k)dk$. Then rewrite $\mathbb{A}$ as an operator on $L^2(\Reals)$ by considering $x,y \in \Reals$ and:
$$\mathbf{A}(x,y)=\chi(x)e^{-h(x+y)}\widehat{A_h}(x+y)\chi(y)$$
where $\chi$ is the characteristic function on $\Reals_+$. By convolution and a change of variable, we have:
\begin{align*}
\widehat{A_h}(x+y) &= \left( \widehat{\sqrt{A_h}} * \widehat{\sqrt{A_h}}\right)(x+y)\\
&= \int_\Reals \widehat{\sqrt{A_h}}(x-s)\widehat{\sqrt{A_h}}(y+s)ds.
\end{align*}

So $\mathbb{A}=\mathbb{A}_1\mathbb{A}_2$ where $\mathbb{A}_1,\mathbb{A}_2$ are operators on $L^2(\Reals)$ with kernels
\begin{align*}
\mathbf{A}_1(x,s) &= \chi(x)e^{-hx}\widehat{\sqrt{A_h}}(x-s),\\
\mathbf{A}_2(s,y) &= \chi(y)e^{-hy}\widehat{\sqrt{A_h}}(y+s).
\end{align*}

One readily has 
\begin{align*}
\norm{\mathbb{A}_k}_{\mathfrak{S}_2}^2 &= \iint_{\Reals^2}\abs{\mathbf{A}_k(\xi,\eta)}^2d\xi d\eta \\
&= \int_\Reals \chi(z)e^{-2hz}dz \int_\Reals \abs{\widehat{\sqrt{A_h}}(S)}^2dS\\
&= \frac{1}{4\pi h}\norm{\sqrt{A_h}}_{L^2(\Reals)}^2=\frac{1}{4\pi h}\norm{A_h}_{L^1(\Reals)}
\end{align*}
where we have used the Plancherel equality $\norm{\widehat{f}}_2^2=\frac{1}{2\pi}\norm{f}_2^2$ and hence
\begin{equation*}\norm{\mathbb{A}}_{\mathfrak{S}_1} \leq \norm{\mathbb{A}_1}_{\mathfrak{S}_2}\norm{\mathbb{A}_2}_{\mathfrak{S}_2}=\frac{1}{4\pi h}\norm{A_h}_{L^1(\Reals)}=\frac{1}{4\pi h}\norm{A}_{L^1(\Reals+ih)}. \qedhere
\end{equation*}
\end{proof}

\section{Classical Marchenko inverse scattering}

In this section we review well-known facts about how to recover a potential $q$ in the Faddeev class from the scattering data associated with the Schr\"odinger operator $-\partial_x^2+q(x)$ via the Marchenko inverse scattering procedure (see e.g. \cite{AK01,Deift79,Fad}).

For $q \in L^1_1(\Reals)$, the scattering data consisting of
\begin{itemize}
\item the discrete spectrum $\{-\varkappa_n^2\}_{n=1}^N$ of the Schr\"{o}dinger operator $-\partial_x^2+q(x)$ on $L^2(\Reals)$,
\item norming constants $\{c_n\}_{n=1}^N$ associated to the bound states of the Schr\"{o}dinger operator,
\item and the reflection coefficient $R(k)$, $k\in\Reals$
\end{itemize}
determine together the potential uniquely. By the inverse scattering procedure 
\begin{equation}\label{ISq}
q(x)= -2 \partial_xk_x(0^+),
\end{equation}
where $k_x \in L^2(\Reals_+)$ solves the Marchenko equation
\begin{equation}\label{kx}
k_x(y)+M_x(y)+ \int_0^\infty M_x(y+z)k_x(z)dz=0 \quad,\quad y>0
\end{equation}
with
\begin{align}
M_x(\cdot) &= M(\cdot+2x) \label{eqMx2},\\
M(s) &=\sum_{n=1}^N c_n^2e^{-\varkappa_ns}+\frac{1}{2\pi}\int_\Reals e^{iks}R(k)dk. \label{eqMker}
\end{align}

If we define the Marchenko operator on $L^2(\Reals_+)$ as:
\begin{equation}
\left(\mathbb{M}_xf\right)(y) = \int_0^\infty M_x(y+s)f(s)ds \quad, \quad f\in L^2(\Reals_+) \label{eqMxop} 
\end{equation}
then \eqref{kx} becomes
\begin{equation*}
(1+\mathbb{M}_x)k_x(y)=-M_x(y)
\end{equation*}
and $1+\mathbb{M}_x$ is boundedly invertible 
\cite{Deift79}.

Assuming the Fredholm determinant in \eqref{eqMxop} is well-defined, one can also rewrite \eqref{ISq} as \cite{Fad}
\begin{equation}\label{Fadq}
q(x) = -2 \partial_x^2 \log \det (1+\mathbb{M}_x),
\end{equation}
known as the Bargmann, Dyson, or determinant formula (see e.g. \cite{RybIP09}). The determinant is well-defined if $\mathbb{M}_x$ is trace class. However, we don't know if it is the case for a generic short range potential.

We choose to detour this fact. To this end, we express the Marchenko kernel in a different form. Recall the well-known fact (see e.g. \cite{AK01}) that if a short range potential is supported on $\Reals_-$, then $R(k)$ can be analytically continued in $\C^+$, its poles are $\{i\varkappa_n\}_{n=1}^N$ and 
$$\Res (R(k),i\varkappa_n)=ic_n^2, $$
where $c_n$ is the norming constant associated to $\varkappa_n$ in the scattering data.

So for any $h> \max\{\varkappa_n\}$, one can deform the contour \cite{RybIP09} in \eqref{eqMker} and by the residue theorem rewrite the Marchenko kernel \eqref{eqMker}
\begin{equation}\label{eqMh}
M(s)= \frac{1}{2\pi} \int_{\Reals+ih}e^{iks}R(k)dk.
\end{equation}

Note that \eqref{eqMh} can then be used for any compactly supported $q$'s with compact support using a shifting argument. 

\begin{remark} We will consider a potential which is locally square integrable on the line, in $\ell^\infty(L^2(\Reals_-))$ and such that $q_+$ is Faddeev class. 
The classical inverse scattering results do not apply directly since $q$ is not short range and the negative part of the spectrum of $-\partial_x^2+q(x)$ need not be finite. However, since $q \in L^2_{loc}(\Reals)\subset L^1_{loc}(\Reals)$, we have 
\begin{itemize}
\item $\tilde{q} \in L^1_1(\Reals)$ so the classical inverse scattering procedure applies to the truncated potential,
\item $\tilde{q}$ is compactly supported so we can use \eqref{eqMh}, 
\item $\tilde{q}\in \ell^\infty(L^2(\Reals))$ which we will show implies that $\widetilde{\mathbb{M}}_x$ is trace class and so \eqref{Fadq} applies.
\end{itemize}
The above will be a basis for our limiting procedure.
\end{remark}

\section{Main result}

We now present our main result which gives a formula to recover a nondecaying unknown potential $q_-$ assuming that $q_+$ and the reflection coefficient $R$ are known.

\begin{theorem}\label{thm8.1} Let $q$ be a real, locally square integrable potential on $\Reals$ such that  ($q_\pm=q|_{\Reals_\pm}$)
\begin{itemize}
\item $\displaystyle\sup_{x\le0} \int_{x-1}^x \abs{q_-(s)}^2ds <\infty$ (uniformly in $L^2_{loc}$),
\item $\displaystyle \int_{\Reals_+}(1+x)\abs{q_+(x)}dx<\infty$ (short range) 
\end{itemize} 
and let $R(k)$, $R_+(k)$ be the right reflection coefficient corresponding to $q,q_+$ respectively. 

Let $\mathbb{M}_x^+$ be the Marchenko operator associated with the scattering data $$\left\{ R_+(k),-(\varkappa_n^+)^2,c_n^+ \right\}_{k\in\Reals, 1\le n\le N}$$ for $q_+$ (given by \eqref{eqMx2}-\eqref{eqMxop}) and let $\mathbb{G}_x$ be the Hankel integral operator associated with $R-R_+$. I.e.
\begin{align}
\left(\mathbb{G}_xf\right)(y) &= \int_0^\infty \mathbf{G}_x(y+s)f(s)ds \quad,\quad f\in L^2(\Reals_+), \label{eq8.-1}\\
\mathbf{G}_x(s) &= \frac{1}{2\pi} \int_{\Reals+ih}e^{ik(s+2x)}(R-R_+)(k)dk \label{eq8.0}
\end{align}
with some $h>0$ sufficiently large.

Then for any $x<0$ 
\begin{equation}\label{eqqm}
q_-(x)=-2 \partial_x^2 \log\det \left(1+\left(1+\mathbb{M}_x^+\right)^{-1}\mathbb{G}_x\right)
\end{equation}
with the determinant defined in the classical Fredholm sense.
\end{theorem}

\begin{remark}Theorem \ref{thm8.1} solves the inverse scattering problem for a steplike potential with the knowledge of its short range part. Indeed, given (short range) $q_+$ one solves the direct scattering problem and finds the scattering data $$\left\{ R_+(k),-(\varkappa_n)^2,c_n \right\}_{k\in\Reals, 1\le n\le N}$$ for $q_+$. Then, we construct $\mathbb{M}_x^+$ by \eqref{eqMx2}-\eqref{eqMxop} and given the (right) reflection coefficient for the whole potential $q$, one constructs by \eqref{eq8.-1}-\eqref{eq8.0} the Hankel operator $\mathbb{G}_x$. The unknown (non decaying) part of $q$ is recovered for each $x<0$ by \eqref{eqqm}. 
\end{remark}

\begin{proof} We will first prove the statement for $\tilde{q}$. For a fixed $a>0$, $\tilde{q}$ is compactly supported. Hence, $\widetilde{R}$ can be analytically continued in $\C^+$ except at a finite number of poles $\{i\widetilde{\varkappa}_n\}_{n=1}^{\widetilde{N}}$, and the Marchenko kernel \eqref{eqMker} becomes
\begin{equation*}
\widetilde{M}(s)= \frac{1}{2\pi}\int_{\Reals+ih}e^{iks}\widetilde{R}(k)dk \quad , \quad h>\max\{\widetilde{\varkappa}_n\}_{n=1}^{\widetilde{N}}.
\end{equation*}
Define $q_a(x)=\tilde{q}(x+a)$ as in Figure \ref{fig4}. Then $q_a$ is supported on $\Reals_-$ and $q_a\in \ell^\infty(L^2(\Reals_-))$.

\begin{figure}[htb]
\begin{tikzpicture}[scale=0.6]

\def \qminus {(-8.3,3.5) to[out=40,in=180] (-7.2,3.3) to[out=0,in=180] (-6.2,3.6) .. controls +(1.5,0) and +(-1,0) .. (-3.8,-1) .. controls +(1,0) and +(-1,0) .. (-2.1,3.7) to[out=0,in=165] (0,2.7) }

\def \qplus { (0,2.7) to[out=-15,in=150] (1.5,2.2) to[out=-30,,in=160] (3.3,1.2) to[out=-20,in=135] (5.3,0.6) to[out=-45,in=180] (6.1,-0.3) to[out=0,in=180] (7.4,0.4) to[out=0,in=180] (10,0) }

\def \a {3.7}

\def \qplusshaded { \qplus -| (0,2.7)}

\draw[help lines] (-8.5,0) -- (10,0)
                  (0,-1) -- +(0,6);
\draw (0,0) node[below left] {$0$};

\draw[dotted] 
\qminus \qplus;

\begin{scope}
\clip (0,-1) rectangle (-2*\a,4);
\draw[thick,blue,xshift=-\a cm]
\qminus \qplus;
\begin{scope}[xshift=-\a cm]
\clip \qplusshaded;
\foreach \y in {-2,-1.6,...,10}
\draw[domain=0:10,help lines] plot (\x,{8/5*(\x-\y)});
\end{scope}
\end{scope}

\foreach \k/\j in {-2/-2,-1/-,1/} 
\draw[help lines]
(\k*\a,0)+(0,-0.2) node[below] {{\j}a} -- +(0,0.2);

\draw (-2*\a,1.3) node[above] {$q_a(x)$};
\draw (2*\a,1.5) node[below] {$q(x)$};

\end{tikzpicture}
\caption{Shifted potential $q_a(x)=\tilde{q}(x+a)$}
\label{fig4}
\end{figure}
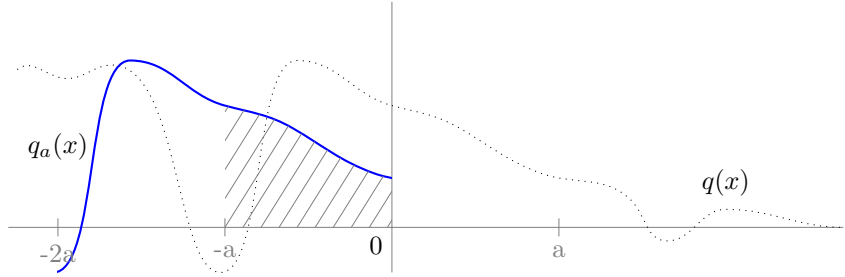

Now its reflection coefficient $R_{a,-}(k) \in L^1(\Reals+ih)$ for $h>h_a$ where $h_a$ is defined as in Proposition \ref{pr4.1} for $q_a$. But $\widetilde{R}(k)=R_{a,-}(k)e^{-2ika}$, so they share the same poles and by Corollary \ref{cor5.3}, $\widetilde{R}(k), \widetilde{R}(k)e^{2ikx}\in L^1(\Reals+ih)$ for $h>h_a$. We also have that $\widetilde{R}(k)e^{2ikx}$ is smooth on $\Reals+ih$  for $h>\max\{\widetilde{\varkappa}_n\}_{n=1}^{\widetilde{N}}$ so by Proposition \ref{pr6.2}, $\widetilde{\mathbb{M}}_x$ is trace class. 
Hence the following Bargmann formula applies:
\begin{equation}\label{BargtM}
\tilde{q}(x)=-2 \partial_x^2\log\det(1+\widetilde{\mathbb{M}}_x).
\end{equation}

Now write
$$\widetilde{R}= \widetilde{R}_+ + \widetilde{G} \quad, \quad \widetilde{G}={\Delta \widetilde{R}}= \widetilde{R}-\widetilde{R}_+$$
and split the Marchenko operator accordingly\footnote{Because $\tilde{q}_+$ is compactly supported, $\widetilde{{M}}^+$ can be equivalently expressed as \eqref{eqMker} or \eqref{eqMh}.}:
$$\widetilde{\mathbb{M}}_x= \widetilde{\mathbb{M}}_x^++ \widetilde{\mathbb{G}}_x. $$
The same $h_a$ is enough to ensure $\widetilde{\mathbb{M}}_x^+\in\mathfrak{S}_1$ since $\widetilde{R}_+(k)e^{2ika}$ corresponds to $\tilde{q}_a(x)$ (above the shaded region in Figure \ref{fig4}). 

In addition, for $h>h_0$ where $h_0$ is the same\footnote{Note that $h_0$ is independent of $a$.} as in Proposition \ref{pr4.3}, $\widetilde{G},G,\delta G \in L^1(\Reals+ih)$ and by Proposition \ref{pr5.5} and the subsequent remark, we also have $\widetilde{G},G,\delta G$ smooth on $\Reals+ih$. Since by Corollary \ref{cor5.8}, the same applies to $\widetilde{G}_x,G_x,\delta G_x$, we can apply Proposition \ref{pr6.2} and conclude that $\widetilde{\mathbb{G}}_x,\mathbb{G}_x,\delta\mathbb{G}_x \in \mathfrak{S}_1$.

Therefore, first we rewrite the Bargmann formula \eqref{BargtM} as:
\begin{align}
\tilde{q}(x) &=-2 \partial_x^2\log\det(1+\widetilde{\mathbb{M}}_x^++\widetilde{\mathbb{G}}_x) \notag \\
&=-2 \partial_x^2\log\det(1+\mathbb{M}_x^+)-2 \partial_x^2\log\det(1+(1+\widetilde{\mathbb{M}}_x^+)^{-1}\widetilde{\mathbb{G}}_x) \label{eqBargtMpG}
\end{align}
where we have used the fact from classical Marchenko theory that $1+\widetilde{\mathbb{M}}_x$ is boundedly invertible. But
\begin{equation*} 
\tilde{q}_+(x)= -2\partial_x^2 \log\det(1+\widetilde{M}_x^+)
\end{equation*}
and $q_+(x)=0$ for $x<0$. So \eqref{eqBargtMpG} becomes for $x<0$:
\begin{equation}\label{eqBargqm}
\tilde{q}_-(x) = -2 \partial_x^2\log\det(1+(1+\widetilde{\mathbb{M}}_x^+)^{-1}\widetilde{\mathbb{G}}_x).
\end{equation}

Now, we now use the fact that $(1+\mathbb{M}_x)$ remains boundedly invertible in its limit so the right hand side of \eqref{eqqm} is well-defined. In addition, by Proposition \ref{pr6.2}, for $h>h_0$
$$\norm{\delta\mathbb{G}_x}_{\mathfrak{S}_1} \leq \frac{1}{4\pi h}\norm{\delta G_x}_{L^1(\Reals+ih)} $$
and the right hand hand side of the inequality goes to zero by Corollary \ref{cor5.8} for $a\to\infty$. Therefore, we find indeed that the limit of \eqref{eqBargqm} is \eqref{eqqm}.
\end{proof}

Note that if $\mathbb{M}_x^+\in\mathfrak{S}_1$ then \eqref{eqqm} simplifies to
\begin{equation}
q(x)= -2\partial_x^2\log\det (1+\mathbb{M}_x),\quad x\in \Reals,
\end{equation}
where $\mathbb{M}_x=\mathbb{M}_x^++\mathbb{G}_x$. It is, of course, well-known (see e.g. \cite{Deift79}) that under our condition on $q_+$, $\mathbb{M}_x \in \mathfrak{S}_2$ but we couldn't prove it for $\mathfrak{S}_1$. We were unable to find a rigorous proof of such a statement in the literature either. (It is typically assumed (frequently implicitly) or referred to as ``too involved".) However, since $\mathbb{M}_x\in\mathfrak{S}_2$, then $\det(1+\mathbb{M}_x)$ can, in fact, be regularized differently from \eqref{eqqm} (see \cite{Ryprep} for details).

In conclusion, we emphasize that the fact that \eqref{eqqm} is understood in the classical sense is indeed quite important as it guarantees the convergence of various types of approximation of $\left(1+\mathbb{M}_x^+\right)^{-1}\mathbb{G}_x$ in trace norm. This, in turn, means a certain stability of the inverse problem algorithm based upon \eqref{eqqm}.

\end{document}